\documentclass[11pt]{article}

\usepackage[margin=1in]{geometry}
\usepackage{amsmath,amssymb,amsthm,mathtools}
\usepackage[numbers]{natbib}
\usepackage{microtype}
\usepackage{enumitem}
\usepackage{hyperref}
\usepackage[nameinlink,capitalize]{cleveref}

\crefname{theorem}{theorem}{theorems}
\Crefname{theorem}{Theorem}{Theorems}
\crefname{lemma}{lemma}{lemmas}
\Crefname{lemma}{Lemma}{Lemmas}
\crefname{proposition}{proposition}{propositions}
\Crefname{proposition}{Proposition}{Propositions}
\crefname{corollary}{corollary}{corollaries}
\Crefname{corollary}{Corollary}{Corollaries}
\crefname{definition}{definition}{definitions}
\Crefname{definition}{Definition}{Definitions}
\crefname{remark}{remark}{remarks}
\Crefname{remark}{Remark}{Remarks}
\crefname{example}{example}{examples}
\Crefname{example}{Example}{Examples}

\hypersetup{
  pdftitle={Learnable Predictions Need Not Be Actionable: Proper Repair Dimension for Online Buying},
  pdfauthor={Yushan Li},
  colorlinks=true,
  linkcolor=blue,
  citecolor=blue,
  urlcolor=blue
}

\allowdisplaybreaks

\newtheorem{theorem}{Theorem}[section]
\newtheorem{lemma}[theorem]{Lemma}

\newtheorem{corollary}[theorem]{Corollary}
\newtheorem{definition}[theorem]{Definition}
\newtheorem{remark}[theorem]{Remark}

\newcommand{\PRD}{\mathsf{PRD}}
\newcommand{\Ldim}{\mathsf{Ldim}}
\newcommand{\cost}{\mathsf{cost}}
\newcommand{\Serve}{\mathsf{Serve}}
\newcommand{\uOPT}{u\mathsf{OPT}}
\newcommand{\GammaD}{\Gamma_D}
\newcommand{\GammaZero}{\Gamma_0}
\newcommand{\E}{\mathsf{E}}
\newcommand{\R}{\mathbb{R}}

\title{Learnable Predictions Need Not Be Actionable: Proper Repair Dimension for Online Buying}
\author{Yushan Li}
\date{}

\begin{document}
\maketitle

\begin{abstract}
Learnability and actionability are different requirements for online predictions. Littlestone dimension controls ordinary online learnability, but proper repair dimension controls actionable online buying. A concept class can be easy to learn in the standard mistake-bound sense and still be hard to maintain as a live actionable prediction once the algorithm is required to stay proper and to buy the realized action online.

For a finite binary concept class, the paper defines the \emph{proper repair dimension} $\PRD(\mathcal{H})$ by a dynamic program on version spaces. The value $\PRD(\mathcal{H})$ is exactly the optimal deterministic worst-case number of repairs for a proper learner that must keep a live hypothesis through every realizable labeled sequence. The paper then proves
\[
  \Ldim(\mathcal{H}) \leq \PRD(\mathcal{H}) \leq |\mathcal{H}|-1.
\]
Both bounds are tight: the full class on $d$ coordinates has $\Ldim=\PRD=d$, while the universal coordinate class $U_n$ has $\Ldim(U_n)=\lfloor \log_2 n \rfloor$ but $\PRD(U_n)=n-1$.

This gap transfers directly to online buying. The paper builds a unit-cost actionable buying instance from every proper class $\mathcal{H}$ and proves that the terminal benchmark is $\uOPT=1$ while every deterministic proper actionable algorithm pays exactly $1+\PRD(\mathcal{H})$ in the worst case. For $U_{2^d}$, this gives deterministic cost $2^d$ despite Littlestone dimension $d$.

The paper also gives a positive transfer theorem for componentized actionable prediction classes. If component $i$ has proper repair dimension $D_i$, service-failure count $I_i$, load $L_i=1+D_i+I_i$, scale distortion $\lambda$, and PRD-load congestion $\GammaD$, then Bellman stage repair satisfies
\[
  \cost \leq \lambda \GammaD \uOPT_C.
\]
Under exact service and the uniform bound $D_i \leq D$, this simplifies to
\[
  \cost \leq \lambda (1+D)\GammaZero \uOPT_C.
\]
\end{abstract}

\section{Introduction}

Learnable predictions need not be actionable. An online algorithm benefits from a prediction only after that prediction has been realized as a live object that can actually serve the arriving requests: a chosen resource, a maintained hypothesis, a bought action, or a currently valid certificate. In many online settings, the prediction must also remain realized as a proper live object while evidence arrives and the feasible set shrinks.

Algorithms with predictions usually study the quality of a predicted object: future requests, a candidate optimal solution, a dual witness, or one predictor among several possibilities \citep{lykouris2020competitive,purohit2024improving,bamas2020primal,anand2022multiple,angelopoulos2024untrusted,mitzenmacher2020algorithms,khodak2022learning}. Those models are informative when the prediction is already actionable. The present paper isolates the missing structural question: when does a learnable prediction class admit a low-repair proper actionable realization online?

The gap already appears in the simplest proper online prediction model. For standard online binary prediction, the classical benchmark is Littlestone dimension \citep{littlestone1988learning}. A small Littlestone dimension means that the class is learnable with a finite deterministic mistake bound. That parameter does not control actionability. If the online algorithm must stay proper, keep a live hypothesis, and buy the action realized by that hypothesis online, then the governing quantity is different.

This paper introduces \emph{proper repair dimension}. For a version space $V$ and live hypothesis $h \in V$, the quantity $D(V,h)$ is the worst-case number of future proper repairs after the adversary chooses a labeled example that eliminates $h$. The state value $D(V)$ minimizes over the choice of the current live hypothesis, and
\[
  \PRD(\mathcal{H}) = D(\mathcal{H})
\]
is the resulting proper repair dimension of the class.

The theorem spine is direct.

\begin{enumerate}[leftmargin=2em]
  \item \textbf{Exact proper repair characterization.} The value $\PRD(\mathcal{H})$ is exactly the optimal deterministic worst-case number of future repairs for a proper learner that must keep a live hypothesis in the current version space.
  \item \textbf{Comparison with Littlestone dimension.} Every finite binary class satisfies
  \[
    \Ldim(\mathcal{H}) \leq \PRD(\mathcal{H}) \leq |\mathcal{H}|-1.
  \]
  \item \textbf{Tight examples.} The full class on $d$ coordinates satisfies $\Ldim=\PRD=d$. The universal coordinate class $U_n$ satisfies
  \[
    \Ldim(U_n)=\lfloor \log_2 n \rfloor,
    \qquad
    \PRD(U_n)=n-1.
  \]
  In particular, when $n=2^d$ the gap is exponential in the Littlestone dimension.
  \item \textbf{Actionable online buying lower bound.} The induced proper actionable buying instance from $\mathcal{H}$ has terminal benchmark $\uOPT=1$ but deterministic proper actionable cost
  \[
    1+\PRD(\mathcal{H}).
  \]
  Hence $U_{2^d}$ yields deterministic cost $2^d$ despite $\Ldim=d$.
  \item \textbf{Positive transfer theorem.} For component classes with loads $L_i=1+D_i+I_i$, scale distortion $\lambda$, and PRD-load congestion $\GammaD$, Bellman stage repair satisfies
  \[
    \cost \leq \lambda \GammaD \uOPT_C.
  \]
\end{enumerate}

The main message is sharp.
\begin{quote}
Small Littlestone dimension guarantees learnability. Small proper repair dimension guarantees actionability.
\end{quote}

PRD is certified by the full chain of results: an exact proper repair characterization, an exponential separation from Littlestone dimension, and an exact online-buying cost formula with terminal benchmark $\uOPT=1$. This chain makes PRD an operational invariant for actionable online buying by tying the version-space recursion to a sharp realizable cost separation.

These are different dimensions because they solve different online tasks. Littlestone dimension measures how many mistakes an online predictor must make before it identifies the hidden label rule. Proper repair dimension measures how many times a proper actionable predictor can be forced to abandon its current live hypothesis while the sequence stays realizable. Once predictions have to be maintained online as bought actions, that second task is the one that governs utility.

Proper repair dimension can be viewed as the unit-cost version-space specialization of a repair dynamic program. This specialization exposes the exact combinatorics needed for the Littlestone comparison and the online-buying reduction, turning the gap into a direct separation between learnability and actionable utility.

\paragraph{Relation to prior work.}
Learning-augmented online algorithms study how predictions improve online competitive guarantees in caching-style problems and many broader settings \citep{lykouris2020competitive,purohit2024improving,bamas2020primal,anand2022multiple,khodak2022learning,azar2021graph,coester2026dual,ameli2025covering,mitzenmacher2020algorithms}. This line builds on the classical competitive-analysis viewpoint for paging and list update \citep{sleator1985amortized}. Robustness to unreliable predictions is a parallel theme, including untrusted-prediction models for online computation and metrical settings \citep{angelopoulos2024untrusted,antoniadis2023online}.

On the online learning side, Littlestone dimension gives the classical mistake-tree characterization of realizable learnability \citep{littlestone1988learning}. Recent work sharpens this picture with randomized Littlestone dimension and computable optimal online learners, and it makes the role of Standard-Optimal-Algorithm-type recursions explicit \citep{filmus2023rldim,hasrati2023computable}.

Advice complexity and partial-information models study how additional trusted or untrusted information changes online competitiveness \citep{emek2011advice,renault2015advice,boyar2016survey,angelopoulos2024untrusted}. Replay-adversary online learning exhibits a proper/improper separation \citep{dmitriev2025replay}. The present paper identifies an actionable online-buying form of this properness distinction: the prediction object is a proper concept class together with an action realization map, the terminal benchmark remains $\uOPT=1$, the worst-case deterministic proper actionable cost becomes $1+\PRD(\mathcal{H})$, and a matching positive transfer theorem shows that bounded repair dimension and bounded PRD-load congestion suffice for strong buying guarantees.

\paragraph{Paper map.}
\Cref{sec:prd} defines proper repair dimension and proves the exact characterization. \Cref{sec:ldim} compares it to Littlestone dimension and proves the tight examples. \Cref{sec:buying-gap} turns the gap into an online buying lower bound. \Cref{sec:transfer} gives the positive component transfer theorem.

\section{Proper Repair Dimension}\label{sec:prd}

Fix a finite domain-independent binary concept class
\[
  \mathcal{H} \subseteq \{0,1\}^X.
\]
For a version space $V \subseteq \mathcal{H}$, a point $x \in X$, and a label $b \in \{0,1\}$, write
\[
  V_{x,b} = \{ g \in V : g(x)=b \}.
\]

\begin{definition}[Proper repair process]
A realizable labeled sequence
\[
  (x_1,b_1),\dots,(x_T,b_T)
\]
generates version spaces
\[
  V_0=\mathcal{H},
  \qquad
  V_t=(V_{t-1})_{x_t,b_t},
  \qquad
  t=1,\dots,T,
\]
with every $V_t$ nonempty.

A \emph{proper repair strategy} chooses a live hypothesis
\[
  h_t \in V_t
\]
for every time $t$. A \emph{repair} occurs at time $t \geq 1$ if $h_t \neq h_{t-1}$.
\end{definition}

\begin{definition}[Legal eliminations and proper repair dimension]
For a nonempty version space $V \subseteq \mathcal{H}$ and a live hypothesis $h \in V$, define the set of legal eliminations
\[
  \E(V,h)=\{V_{x,b} : V_{x,b}\neq\emptyset \text{ and } h(x)\neq b\}.
\]
Define
\[
  D(V,h)=
  \begin{cases}
    0, & \E(V,h)=\emptyset,\\[2mm]
    \displaystyle \max_{W\in \E(V,h)} \bigl(1+D(W)\bigr), & \text{otherwise,}
  \end{cases}
\]
and
\[
  D(V)=\min_{h \in V} D(V,h).
\]
The \emph{proper repair dimension} of the class is
\[
  \PRD(\mathcal{H}) = D(\mathcal{H}).
\]
\end{definition}

The recursion charges only the next time the current live hypothesis is actually falsified. Valid labeled examples that shrink the version space while leaving the live hypothesis correct cost nothing immediately. Their only effect is to change which future legal eliminations remain available.

\begin{lemma}[Restriction monotonicity]\label{lem:restriction}
If $V \subseteq \mathcal{H}$ is nonempty and $W=V_{x,b}$ is a nonempty one-step restriction of $V$, then
\[
  D(W)\leq D(V).
\]
\end{lemma}

\begin{proof}
We induct on $|V|$. The claim is trivial when $|V|=1$.

Fix a nonempty proper restriction $W=V_{x,b}$ and let $h^\star \in V$ satisfy
\[
  D(V)=D(V,h^\star).
\]
If $h^\star(x)\neq b$, then $W \in \E(V,h^\star)$, so
\[
  D(V)=D(V,h^\star)\geq 1+D(W) > D(W).
\]

Now assume $h^\star(x)=b$, so $h^\star \in W$. Let $Y$ be any legal elimination in $\E(W,h^\star)$. Then
\[
  Y=W_{y,c}
\]
for some $y \in X$ and $c \neq h^\star(y)$. Let
\[
  U=V_{y,c}.
\]
Because $h^\star(y)\neq c$, the set $U$ is a legal elimination in $\E(V,h^\star)$, and
\[
  Y = U_{x,b}.
\]
Since $|U|<|V|$, the induction hypothesis gives
\[
  D(Y)\leq D(U).
\]
Therefore
\[
  D(W)
  \leq D(W,h^\star)
  = \max_{Y\in \E(W,h^\star)} \bigl(1+D(Y)\bigr)
  \leq \max_{U\in \E(V,h^\star)} \bigl(1+D(U)\bigr)
  = D(V,h^\star)
  = D(V).
\]
\end{proof}

\begin{lemma}[Pair monotonicity under surviving restrictions]\label{lem:pair-monotone}
If $W=V_{x,h(x)}$ is a nonempty one-step restriction that keeps $h \in V$ live, then
\[
  D(W,h)\leq D(V,h).
\]
\end{lemma}

\begin{proof}
If $\E(W,h)=\emptyset$, the claim is immediate. Otherwise every $Y \in \E(W,h)$ has the form
\[
  Y=W_{y,c}
\]
with $c \neq h(y)$. Let
\[
  U=V_{y,c}.
\]
Then $U \in \E(V,h)$ and
\[
  Y = U_{x,h(x)}.
\]
By Lemma~\ref{lem:restriction},
\[
  D(Y)\leq D(U).
\]
Taking the maximum over all legal eliminations of $h$ from $W$ gives
\[
  D(W,h)\leq D(V,h).
\]
\end{proof}

\begin{theorem}[Exact proper repair characterization]\label{thm:exact-prd}
For every finite binary concept class $\mathcal{H}$,
\[
  \PRD(\mathcal{H})
\]
is exactly the optimal deterministic worst-case number of repairs for a proper learner that must keep a live hypothesis in the current version space.
\end{theorem}

\begin{proof}
We prove the stronger pair-level statement: for every nonempty version space $V$ and live hypothesis $h \in V$, the value $D(V,h)$ is the optimal deterministic worst-case number of future repairs from the pair $(V,h)$. The theorem follows by minimizing over the initial choice of $h \in \mathcal{H}$.

We induct on $|V|$. If $\E(V,h)=\emptyset$, no valid future labeled example can eliminate $h$, so the optimal future repair count is $0=D(V,h)$.

For the upper bound, use the Bellman strategy: whenever the current live hypothesis is eliminated and the new version space is $W$, switch to a minimizer of $D(W,g)$ over $g \in W$.

Fix a continuation from $(V,h)$ and let $(x,b)$ be the first labeled example that eliminates $h$, if any exists. Let $W$ be the version space after that example. Every earlier labeled example is correct on $h$, so starting from the one-step legal elimination
\[
  U=V_{x,b} \in \E(V,h),
\]
the same earlier labels simply realize $W$ as a sequence of further nonempty restrictions of $U$. Repeated applications of Lemma~\ref{lem:restriction} give
\[
  D(W)\leq D(U).
\]
After repairing at $W$ to a Bellman minimizer, the induction hypothesis bounds the remaining future repairs by $D(W)$. Hence the total future repairs from $(V,h)$ on this continuation are at most
\[
  1 + D(W)
  \leq
  1 + D(U)
  \leq
  \max_{U' \in \E(V,h)} \bigl(1+D(U')\bigr)
  =
  D(V,h).
\]

For the lower bound, let
\[
  U^\star \in \arg\max_{U \in \E(V,h)} \bigl(1+D(U)\bigr).
\]
The adversary reveals the single labeled example witnessing $U^\star \in \E(V,h)$ immediately. The learner must repair once and choose some new live hypothesis $g \in U^\star$. By the induction hypothesis, the adversary can then force at least $D(U^\star,g)\geq D(U^\star)$ additional repairs. Thus every deterministic strategy from $(V,h)$ suffers at least
\[
  1 + D(U^\star) = D(V,h)
\]
future repairs.

The pair-level claim follows. At the initial state, the learner chooses some $h_0 \in \mathcal{H}$ and then faces future repair cost $D(\mathcal{H},h_0)$. Therefore the optimal deterministic worst-case repair count is
\[
  \min_{h_0 \in \mathcal{H}} D(\mathcal{H},h_0)=D(\mathcal{H})=\PRD(\mathcal{H}).
\]
\end{proof}

\begin{remark}
Proper repair dimension is the unit-cost version-space specialization of repair value. The dynamic-programming viewpoint is the same; the present recursion strips away general certificate weights and exposes the exact proper combinatorics.
\end{remark}

\section{Littlestone Dimension and Tight Gaps}\label{sec:ldim}

\begin{definition}[Recursive Littlestone dimension]
For every nonempty finite version space $V \subseteq \mathcal{H}$, define $\Ldim(V)$ recursively by
\[
  \Ldim(V)=0
\]
if for every point $x \in X$, at most one of $V_{x,0}$ and $V_{x,1}$ is nonempty. Otherwise
\[
  \Ldim(V)
  =
  1+\max_{x : V_{x,0}\neq\emptyset,\; V_{x,1}\neq\emptyset}
  \min\bigl\{\Ldim(V_{x,0}),\Ldim(V_{x,1})\bigr\}.
\]
This is the usual Littlestone dimension of the class, specialized to finite version spaces \citep{littlestone1988learning}.
\end{definition}

\begin{theorem}[Proper repair dimension versus Littlestone dimension]\label{thm:ldim-prd}
Every finite binary concept class $\mathcal{H}$ satisfies
\[
  \Ldim(\mathcal{H}) \leq \PRD(\mathcal{H}) \leq |\mathcal{H}|-1.
\]
\end{theorem}

\begin{proof}
We prove the upper bound first. More generally, for every nonempty version space $V$,
\[
  D(V)\leq |V|-1.
\]
Induct on $|V|$. The statement is trivial for $|V|=1$. If $h \in V$ and $W \in \E(V,h)$, then $W$ is a nonempty proper subset of $V$, so the induction hypothesis gives
\[
  D(W)\leq |W|-1 \leq |V|-2.
\]
Hence
\[
  D(V,h)\leq |V|-1
\]
for every $h \in V$, and minimizing over $h$ yields $D(V)\leq |V|-1$. Applying this to $V=\mathcal{H}$ proves $\PRD(\mathcal{H})\leq |\mathcal{H}|-1$.

For the lower bound, we show by induction on $|V|$ that
\[
  \Ldim(V)\leq D(V)
\]
for every nonempty version space $V$. If $\Ldim(V)=0$, there is nothing to prove. Otherwise choose a point $x$ that attains the maximum in the recursive definition:
\[
  \Ldim(V)
  =
  1+\min\bigl\{\Ldim(V_{x,0}),\Ldim(V_{x,1})\bigr\}.
\]
Both branches are nonempty. Fix any $h \in V$ and set
\[
  b = 1-h(x).
\]
Then $V_{x,b}\in \E(V,h)$. By induction,
\[
  D(V_{x,b})\geq \Ldim(V_{x,b}).
\]
Therefore
\[
  D(V,h)
  \geq
  1 + D(V_{x,b})
  \geq
  1 + \Ldim(V_{x,b})
  \geq
  1 + \min\bigl\{\Ldim(V_{x,0}),\Ldim(V_{x,1})\bigr\}
  =
  \Ldim(V).
\]
Since this holds for every $h \in V$, we obtain $D(V)\geq \Ldim(V)$. Applying it to $V=\mathcal{H}$ gives
\[
  \PRD(\mathcal{H})=D(\mathcal{H})\geq \Ldim(\mathcal{H}).
\]
\end{proof}

\begin{theorem}[Full class equality]\label{thm:full-class}
Let
\[
  C_d = \{0,1\}^{[d]}.
\]
Then
\[
  \Ldim(C_d)=\PRD(C_d)=d.
\]
\end{theorem}

\begin{proof}
For a partial assignment $\sigma:J \to \{0,1\}$, let
\[
  C(\sigma)=\{h \in C_d : h|_J=\sigma\}.
\]
Write $q=d-|J|$ for the number of free coordinates.

We first prove by induction on $q$ that
\[
  D(C(\sigma))=q.
\]
When $q=0$, the class is a singleton and the value is $0$. Assume $q\geq 1$ and fix any $h \in C(\sigma)$. Every legal elimination of $h$ fixes some free coordinate $j \notin J$ to the opposite label $1-h(j)$, and the resulting child class is exactly a set of the form
\[
  C(\sigma \cup \{j \mapsto 1-h(j)\})
\]
with $q-1$ free coordinates. By induction, every such child has value $q-1$. Hence
\[
  D(C(\sigma),h)=1+(q-1)=q
\]
for every $h \in C(\sigma)$, and therefore $D(C(\sigma))=q$. Taking $J=\emptyset$ gives
\[
  \PRD(C_d)=D(C_d)=d.
\]

For Littlestone dimension, the same free-coordinate recursion holds. If $q\geq 1$, querying any free coordinate produces two children that are again full classes on $q-1$ free coordinates, so
\[
  \Ldim(C(\sigma))=1+(q-1)=q.
\]
Applying this at $\sigma=\emptyset$ yields $\Ldim(C_d)=d$.
\end{proof}

\begin{theorem}[Universal coordinate class]\label{thm:universal-coordinate}
For $n \geq 1$, let
\[
  U_n=\{h_1,\dots,h_n\}\subseteq \{0,1\}^{2^{[n]}}
\]
be the class over the domain $2^{[n]}$ defined by
\[
  h_i(S)=1 \quad \Longleftrightarrow \quad i \in S.
\]
Then
\[
  \PRD(U_n)=n-1
  \qquad\text{and}\qquad
  \Ldim(U_n)=\lfloor \log_2 n \rfloor.
\]
\end{theorem}

\begin{proof}
For any nonempty index set $I \subseteq [n]$, write
\[
  U_I=\{h_i : i \in I\}.
\]
We first prove
\[
  D(U_I)=|I|-1
\]
by induction on $|I|$. The upper bound is immediate from \Cref{thm:ldim-prd}. For the matching lower bound, fix $i \in I$ and choose the query
\[
  S=I\setminus\{i\}
\]
with label $1$. Then
\[
  h_i(S)=0
\]
while every $h_j$ for $j \in I\setminus\{i\}$ satisfies $h_j(S)=1$. Therefore
\[
  (U_I)_{S,1}=U_{I\setminus\{i\}} \in \E(U_I,h_i).
\]
By induction,
\[
  D(U_I,h_i)\geq 1+D(U_{I\setminus\{i\}})=1+(|I|-2)=|I|-1.
\]
Since $i$ was arbitrary, $D(U_I)\geq |I|-1$, and thus $D(U_I)=|I|-1$. Taking $I=[n]$ yields
\[
  \PRD(U_n)=n-1.
\]

For the Littlestone dimension, the upper bound is a leaf count. A complete mistake tree of depth $d$ has $2^d$ root-to-leaf label sequences, and each sequence must be realized by a different concept in $U_n$. Hence
\[
  2^d \leq n,
\]
so
\[
  \Ldim(U_n)\leq \lfloor \log_2 n \rfloor.
\]

For the lower bound, let
\[
  d=\lfloor \log_2 n \rfloor.
\]
We prove by induction on $m=|I|$ that
\[
  \Ldim(U_I)\geq \lfloor \log_2 m \rfloor.
\]
When $m=1$, the claim is trivial. Now let $m\geq 2$ and partition $I$ into two disjoint nonempty sets
\[
  I_0 \cup I_1 = I
\]
with
\[
  |I_0|,|I_1| \geq 2^{\lfloor \log_2 m \rfloor -1}.
\]
Query the point
\[
  S=I_1.
\]
Then every $h_i$ with $i \in I_1$ labels $S$ by $1$, and every $h_i$ with $i \in I_0$ labels $S$ by $0$, so the two children are exactly
\[
  (U_I)_{S,0}=U_{I_0},
  \qquad
  (U_I)_{S,1}=U_{I_1}.
\]
By induction,
\[
  \Ldim(U_{I_0}),\Ldim(U_{I_1}) \geq \lfloor \log_2 m \rfloor - 1.
\]
Therefore
\[
  \Ldim(U_I)\geq 1+\min\{\Ldim(U_{I_0}),\Ldim(U_{I_1})\}\geq \lfloor \log_2 m \rfloor.
\]
Applying this to $I=[n]$ completes the proof.
\end{proof}

\begin{corollary}[Exponential gap]\label{cor:exp-gap}
For every integer $d \geq 0$, the class $U_{2^d}$ satisfies
\[
  \Ldim(U_{2^d})=d
  \qquad\text{and}\qquad
  \PRD(U_{2^d})=2^d-1.
\]
\end{corollary}

The class $U_n$ keeps ordinary online learnability logarithmic while forcing linear proper-actionable repairs. Its obstruction lies in maintaining a proper realized representative as evidence arrives.

\section{Learnable Predictions Need Not Be Actionable}\label{sec:buying-gap}

We now convert proper repair dimension into an online buying lower bound.

\begin{definition}[Induced proper actionable buying instance]
Given a finite binary class $\mathcal{H}$, define a unit-cost online buying instance as follows.
\begin{enumerate}[leftmargin=2em]
  \item The resource set is
  \[
    A_{\mathcal{H}}=\{a_h : h \in \mathcal{H}\},
  \]
  with unit costs $c(a_h)=1$.
  \item For every labeled example $(x,b)$, the service family is
  \[
    \Serve(x,b)=\{S \subseteq A_{\mathcal{H}} : S \cap \{a_h : h(x)=b\}\neq\emptyset\}.
  \]
\end{enumerate}
\end{definition}

The service family is upward closed, so this is a monotone online buying problem. The prediction-side state is still the version space
\[
  V_t=(V_{t-1})_{x_t,b_t}.
\]

\begin{definition}[Proper actionable algorithm]
A \emph{proper actionable algorithm} for the induced instance maintains:
\begin{enumerate}[leftmargin=2em]
  \item a monotone bought set
  \[
    A_0 \subseteq A_1 \subseteq \cdots \subseteq A_T \subseteq A_{\mathcal H},
  \]
  \item a live hypothesis
  \[
    h_t \in V_t
  \]
  such that
  \[
    a_{h_t}\in A_t
  \]
  for every time $t$.
\end{enumerate}
Its online buying cost is the monotone union cost
\[
  \cost(A_T)=|A_T|.
\]
\end{definition}

In this induced instance, every live feasible hypothesis serves the current labeled request exactly: if $h_t \in V_t$, then by definition
\[
  h_t(x_t)=b_t,
\]
so
\[
  a_{h_t}\in \Serve(x_t,b_t).
\]
Thus the actionable requirement is exactly the requirement to keep a bought live proper hypothesis online. Switching back to an already bought singleton would add no new cost, so the lower-bound argument must show that the adversarial worst-case repair path eliminates the current live hypothesis permanently at every forced repair.

\begin{theorem}[Proper actionable buying cost equals $1+\PRD$]\label{thm:buying-gap}
For every finite binary concept class $\mathcal{H}$:
\begin{enumerate}[leftmargin=2em]
  \item on every realizable labeled sequence, the terminal benchmark satisfies
  \[
    \uOPT = 1;
  \]
  \item the optimal deterministic worst-case proper actionable buying cost over realizable labeled sequences is exactly
  \[
    1+\PRD(\mathcal{H}).
  \]
\end{enumerate}
\end{theorem}

\begin{proof}
Fix any realizable labeled sequence. The terminal version space is nonempty by realizability. Choosing any final surviving hypothesis $h^\star \in V_T$ gives the singleton realized set $\{a_{h^\star}\}$ of cost one, so the terminal benchmark is
\[
  \uOPT=1.
\]

We now turn to worst-case deterministic proper actionable cost.

For the upper bound, run the Bellman strategy from \Cref{thm:exact-prd}. Whenever the strategy activates a live hypothesis $h_t$, buy the singleton resource $a_{h_t}$ if it has not been bought before. Since every live feasible hypothesis serves the current labeled request exactly, the strategy remains proper and actionable. The Bellman strategy makes at most $\PRD(\mathcal{H})$ repairs in the worst case, so it activates at most $1+\PRD(\mathcal{H})$ live hypotheses. Hence its monotone union cost is at most
\[
  1+\PRD(\mathcal{H}).
\]

For the lower bound, fix any deterministic proper actionable algorithm and run the recursive lower-bound adversary from \Cref{thm:exact-prd}. This adversary produces a realizable labeled sequence on which the algorithm's current live hypothesis is forced to be eliminated at least
\[
  \PRD(\mathcal{H})
\]
times. Let
\[
  h^{(0)},h^{(1)},\dots,h^{(\PRD(\mathcal{H})-1)}
\]
be the first $\PRD(\mathcal{H})$ live hypotheses eliminated at those forced repairs. Because version spaces only shrink, each eliminated live hypothesis is permanently removed from the version space, so these hypotheses are pairwise distinct. Terminal realizability leaves a final live survivor
\[
  h^{(\PRD(\mathcal{H}))}\in V_T
\]
that is distinct from all eliminated hypotheses. Whenever one of these live hypotheses is active, actionability requires its singleton resource to lie in the monotone bought set. Therefore the final bought set contains the distinct singletons
\[
  a_{h^{(0)}},a_{h^{(1)}},\dots,a_{h^{(\PRD(\mathcal{H}))}},
\]
so its monotone union cost on this adversarial sequence is at least
\[
  1+\PRD(\mathcal{H}).
\]
Voluntary switches or reactivations of already bought still-feasible hypotheses cannot reduce this bound, because the count already comes from the pairwise distinct forced-eliminated live hypotheses together with the final live survivor.

Combining the upper and lower bounds proves that the optimal deterministic worst-case proper actionable buying cost over realizable labeled sequences is exactly
\[
  1+\PRD(\mathcal{H}).
\]
\end{proof}

\begin{corollary}[Learnable predictions need not be actionable]\label{cor:learnable-not-actionable}
For the universal coordinate class $U_{2^d}$, the induced proper actionable buying instance satisfies
\[
  \uOPT=1,
  \qquad
  \Ldim(U_{2^d})=d,
  \qquad
  \text{deterministic proper actionable cost}=2^d.
\]
\end{corollary}

\begin{proof}
Combine \Cref{cor:exp-gap,thm:buying-gap}.
\end{proof}

\section{Component Classes and Actual-Cost Transfer}\label{sec:transfer}

The component model gives a resource-bundle interpretation of actionable predictions. Each component is a subtask with a finite proper hypothesis class; a live hypothesis realizes a resource bundle $\psi_i(h)$; arriving labels shrink the component version space; and buying requests require the current bundle to serve. Shared resources across components are counted once by the terminal union benchmark. Scale distortion controls the cost of one realized bundle, while PRD-load congestion controls how often shared resources can be charged across repaired components.

\begin{definition}[Monotone online buying problem]
A \emph{monotone online buying problem} is a tuple
\[
  P=(A,c,\mathcal{R},\Serve),
\]
where $A$ is a finite resource set, $c:A\to \R_{>0}$ assigns positive costs, $\mathcal{R}$ is the request universe, and every service family
\[
  \Serve(r)\subseteq 2^A
\]
is upward closed.
\end{definition}

\begin{definition}[Componentized actionable prediction class]
Fix a monotone online buying problem $P=(A,c,\mathcal{R},\Serve)$. A \emph{componentized actionable prediction class} consists of finite binary classes
\[
  \mathcal{H}_i \subseteq \{0,1\}^{X_i},
  \qquad
  i \in [K],
\]
together with realization maps
\[
  \psi_i:\mathcal{H}_i \to 2^A.
\]

At each time $t$, the request carries a component label $\ell_t \in [K]$, a labeled example $(x_t,b_t)$ with $x_t \in X_{\ell_t}$ and $b_t \in \{0,1\}$, and a buying request $r_t \in \mathcal{R}$. Only component $\ell_t$ updates at time $t$.

Component $i$ starts from
\[
  V_i^0=\mathcal{H}_i
\]
and updates by
\[
  V_i^t=(V_i^{t-1})_{x_t,b_t}
\]
whenever $\ell_t=i$. The execution is \emph{valid} if every updated version space remains nonempty.
\end{definition}

\begin{definition}[Serving hypotheses and Bellman-compatible service]
When request $t$ updates component $i$, define the serving set
\[
  G_i^t=\{h \in V_i^t : \psi_i(h)\in \Serve(r_t)\}.
\]
The execution has \emph{Bellman-compatible service} if for every updated component $i$ and every valid time $t$,
\[
  G_i^t \cap \arg\min_{h \in V_i^t} D(V_i^t,h)\neq\emptyset.
\]
\end{definition}

\begin{definition}[Bellman stage repair and service failures]
Bellman stage repair maintains a global monotone bought resource set and, for each active component $i$, a current live hypothesis $h_i^t$.

At the first request to component $i$, the algorithm activates a hypothesis
\[
  h_i^t \in G_i^t \cap \arg\min_{h \in V_i^t} D(V_i^t,h).
\]
It then adds the realized set $\psi_i(h_i^t)$ to the global bought resource set.
At any later request to component $i$:
\begin{enumerate}[leftmargin=2em]
  \item if the current live hypothesis $h_i^{t-1}$ still lies in $G_i^t$, keep it;
  \item otherwise activate a new hypothesis
  \[
    h_i^t \in G_i^t \cap \arg\min_{h \in V_i^t} D(V_i^t,h).
  \]
  It then adds $\psi_i(h_i^t)$ to the global bought resource set.
\end{enumerate}

A \emph{service failure} for component $i$ is an activation at time $t$ such that the previous live hypothesis remains feasible:
\[
  h_i^{t-1} \in V_i^t
\]
but does not serve the current request. Let $I_i$ be the number of service failures of component $i$.
\end{definition}

\begin{remark}
An activation need not create new online buying cost: some or all of the resources in $\psi_i(h_i^t)$ may already have been bought earlier by the same component or another component. The load theorem below counts activations, while the transfer theorem converts their multiplicity accounting into an upper bound on the actual monotone union cost.
\end{remark}

\begin{theorem}[PRD load bound]\label{thm:load}
Fix a valid execution with Bellman-compatible service. For every active component $i$, let
\[
  D_i=\PRD(\mathcal{H}_i)=D(\mathcal{H}_i).
\]
Then Bellman stage repair makes at most
\[
  L_i=1+D_i+I_i
\]
activations in component $i$.
\end{theorem}

\begin{proof}
Fix an active component $i$. After every activation, let
\[
  p_t = D(V_i^t,h_i^t)
\]
be the Bellman potential of the current live hypothesis.

At the first activation in component $i$, the algorithm chooses a Bellman minimizer in the current version space, so
\[
  p_t = D(V_i^t)\leq D(\mathcal{H}_i)=D_i.
\]

Now consider any later update time $t$ for component $i$.

If no activation occurs, the component contributes nothing to the count.

If a service-failure activation occurs, then the previous hypothesis remains feasible:
\[
  h_i^{t-1}\in V_i^t.
\]
The algorithm switches to a Bellman minimizer in $V_i^t$, so
\[
  p_t = D(V_i^t)\leq D(V_i^t,h_i^{t-1}).
\]
Because the label update kept $h_i^{t-1}$ live, Lemma~\ref{lem:pair-monotone} gives
\[
  D(V_i^t,h_i^{t-1})\leq D(V_i^{t-1},h_i^{t-1})=p_{t-1}.
\]
Hence a service-failure activation does not increase the potential.

If an activation occurs because the previous live hypothesis is eliminated, then
\[
  V_i^t \in \E(V_i^{t-1},h_i^{t-1}).
\]
Since the new hypothesis is a Bellman minimizer in $V_i^t$,
\[
  p_t = D(V_i^t).
\]
By the recursion for $D(V_i^{t-1},h_i^{t-1})$,
\[
  p_{t-1}
  =
  D(V_i^{t-1},h_i^{t-1})
  \geq
  1 + D(V_i^t)
  =
  1+p_t.
\]
Thus every elimination-driven activation decreases the potential by at least one.

The potential starts at most $D_i$, never increases, and drops by at least one at every elimination-driven activation. Therefore there are at most $D_i$ elimination-driven activations after the first activation. Adding the initial activation and the $I_i$ service-failure activations yields
\[
  L_i=1+D_i+I_i
\]
activations in component $i$.
\end{proof}

\begin{definition}[Terminal deduplicated benchmark]
At terminal time $T$, define
\[
  \uOPT_C
  =
  \min_{h_i \in V_i^T\ \forall i}
  \cost\!\Bigl(\bigcup_{i=1}^K \psi_i(h_i)\Bigr),
\]
with $\uOPT_C=+\infty$ if some terminal version space is empty.
\end{definition}

\begin{definition}[Scale distortion]
Component $i$ has \emph{scale distortion} $(\alpha_i,\lambda_i)$ if $\alpha_i>0$, $\lambda_i\geq 1$, and:
\begin{enumerate}[leftmargin=2em]
  \item every Bellman stage-repair activation in component $i$ realizes resources of total cost at most $\lambda_i \alpha_i$;
  \item every terminal feasible hypothesis $h_i \in V_i^T$ satisfies
  \[
    \cost(\psi_i(h_i)) \geq \alpha_i.
  \]
\end{enumerate}
Let
\[
  \lambda=\max_i \lambda_i.
\]
\end{definition}

\begin{definition}[PRD-load congestion]
Given component loads $L_i$, define the \emph{PRD-load congestion}
\[
  \GammaD
  =
  \max_{a \in A}
  \sum_{i : \exists h \in \mathcal{H}_i,\ a \in \psi_i(h)} L_i.
\]
Define the unweighted overlap parameter
\[
  \GammaZero
  =
  \max_{a \in A}
  \bigl|\{i : \exists h \in \mathcal{H}_i,\ a \in \psi_i(h)\}\bigr|.
\]
\end{definition}

\begin{theorem}[PRD transfer theorem]\label{thm:transfer}
Consider a valid execution of a componentized actionable prediction class under Bellman stage repair. Assume:
\begin{enumerate}[leftmargin=2em]
  \item every active component satisfies Bellman-compatible service;
  \item component $i$ has proper repair dimension $D_i$, service-failure count $I_i$, and load
  \[
    L_i=1+D_i+I_i;
  \]
  \item component $i$ has scale distortion $(\alpha_i,\lambda_i)$;
  \item $\uOPT_C<+\infty$.
\end{enumerate}
Then the actual bought-resource cost satisfies
\[
  \cost \leq \lambda \GammaD \uOPT_C.
\]
If $\uOPT_C=+\infty$, the statement is vacuous.
\end{theorem}

The singleton lower bound proves PRD necessity in the sharpest unit-cost embedding; \Cref{thm:transfer} gives the positive side once realized bundles have bounded scale and bounded overlap.

\begin{proof}
By \Cref{thm:load}, component $i$ makes at most $L_i$ activations. Every such activation realizes resources of total cost at most $\lambda_i \alpha_i \leq \lambda \alpha_i$. Since the actual monotone union cost is bounded by this activation multiplicity, we get
\[
  \cost \leq \lambda \sum_{i=1}^K L_i \alpha_i.
\]

Fix any terminal selector $h_i \in V_i^T$ for every component. By the lower half of the scale-distortion condition,
\[
  \alpha_i \leq \cost(\psi_i(h_i)).
\]
Therefore
\[
  \cost \leq \lambda \sum_{i=1}^K L_i \cost(\psi_i(h_i)).
\]

Expand the right-hand side resource by resource:
\[
  \sum_{i=1}^K L_i \cost(\psi_i(h_i))
  =
  \sum_{a \in A} c(a)\sum_{i : a \in \psi_i(h_i)} L_i.
\]
By the definition of $\GammaD$, every inner sum is at most $\GammaD$, so
\[
  \sum_{i=1}^K L_i \cost(\psi_i(h_i))
  \leq
  \GammaD \sum_{a \in \cup_i \psi_i(h_i)} c(a)
  =
  \GammaD \cost\!\Bigl(\bigcup_{i=1}^K \psi_i(h_i)\Bigr).
\]
Minimizing over all terminal selectors proves
\[
  \cost \leq \lambda \GammaD \uOPT_C.
\]
\end{proof}

\begin{corollary}[Uniform corollary]\label{cor:uniform}
Assume every active component has exact service, meaning that every live feasible hypothesis in the current version space serves every valid current request. Then $I_i=0$, and if in addition every component satisfies
\[
  D_i=\PRD(\mathcal{H}_i)\leq D.
\]
Then Bellman stage repair satisfies
\[
  \cost \leq \lambda (1+D)\GammaZero \uOPT_C.
\]
\end{corollary}

\begin{proof}
Exact service gives $I_i=0$, so \Cref{thm:load} yields
\[
  L_i \leq 1+D
\]
for every active component. Hence
\[
  \GammaD
  =
  \max_{a \in A}
  \sum_{i : \exists h \in \mathcal{H}_i,\ a \in \psi_i(h)} L_i
  \leq
  (1+D)\GammaZero.
\]
Applying \Cref{thm:transfer} gives the result.
\end{proof}

\section{Discussion}

Proper repair dimension isolates the online burden that remains invisible to Littlestone dimension. Littlestone dimension classifies ordinary online learnability; PRD classifies the proper repair burden of maintaining a live realized representative; actionable online buying is governed by PRD in the singleton embedding; and bounded PRD together with bounded realization congestion gives the positive transfer guarantee for componentized buying.

The natural next directions are randomized repair dimensions, improper actionable realizations, and structural design rules that minimize PRD-load congestion without losing predictive expressiveness.

\bibliographystyle{plainnat}
\bibliography{refs}

\end{document}